\documentclass[10pt,journal,compsoc]{IEEEtran}

\usepackage[T1]{fontenc}
\usepackage[utf8]{inputenc}
\usepackage{amsmath,amssymb,amsthm}
\usepackage{enumitem}
\usepackage[expansion=false]{microtype}
\usepackage{graphicx}
\usepackage{xurl}
\usepackage[hidelinks]{hyperref}
\hypersetup{
	pdfauthor={Homayoun Maleki, Nekane Sainz, Jon Legarda, Igor Santos-Grueiro},
	pdftitle={Identity Amortization: Shared Resources and the Cost of Sybil Attacks},
	pdfsubject={Computer Science - Cryptography and Security},
	pdfkeywords={Identity amortization, Sybil resistance, cross-identity resource sharing, shared provisioning, influence allocation, Sybil attacks}
}

\theoremstyle{definition}
\newtheorem{definition}{Definition}
\theoremstyle{plain}
\newtheorem{proposition}{Proposition}
\newtheorem{lemma}{Lemma}
\newtheorem{corollary}{Corollary}
\theoremstyle{remark}

\newcommand{\M}{\mathcal{M}}
\newcommand{\Cshared}{C_{\mathrm{shared}}}
\newcommand{\Cind}{C_{\mathrm{ind}}}
\newcommand{\Cbshared}{\bar C_{\mathrm{shared}}}
\newcommand{\Cbind}{\bar C_{\mathrm{ind}}}
\newcommand{\Pishared}{\Pi_{\mathrm{shared}}}
\newcommand{\Piind}{\Pi_{\mathrm{ind}}}

\title{Identity Amortization: Shared Resources and the Cost of Sybil Attacks}

\author{Homayoun~Maleki,
	Nekane~Sainz,
	Jon~Legarda,
	and~Igor~Santos-Grueiro%
	\IEEEcompsocitemizethanks{%
		\IEEEcompsocthanksitem This work has been submitted to the IEEE for possible publication. Copyright may be transferred without notice, after which this version may no longer be accessible.
		\IEEEcompsocthanksitem H. Maleki, N. Sainz, and J. Legarda are with DeustoTech, Faculty of Engineering, University of Deusto, Bilbao, Spain.\protect\\
		E-mail: h.maleki@deusto.es, nekane.sainz@deusto.es, jlegarda@deusto.es
		\IEEEcompsocthanksitem I. Santos-Grueiro is with the International University of La Rioja (UNIR), Logro\~{n}o, Spain.
		\IEEEcompsocthanksitem Corresponding author: Homayoun Maleki.}}

\IEEEtitleabstractindextext{%
	\begin{abstract}
		Sybil defenses commonly price identities by requiring a resource for each one. When one resource can support several identities, an adversary can pay less per identity than the defense intends. We call this \emph{identity amortization}: a strict reduction in the minimum cost of provisioning identities when resources may be shared, relative to provisioning each one independently. It is a second route to the gap that a Sybil attack exploits. When influence is assigned per identity, a defense calibrated to the nominal price can overestimate the cost of a Sybil threshold by up to the number of identities one resource supports. For any resource catalogue, we give exact tests for when identity amortization occurs, where it begins, and its largest error; a sharing cap below the onset removes it. Because the tests need only prices and capacities, they apply to any defense that admits identities against a shareable resource, telling designers whether the assumed per-identity cost is paid, how wrong it can be, and which intervention restores it most cheaply for honest participants. On published phone-verification prices, per-account estimates overstate provisioning cost 1.6--4.3 times; Tor's raised per-address relay limit cut modeled attack cost by 8--72\%, depending on the address-to-relay price ratio.
	\end{abstract}
	
	\begin{IEEEkeywords}
		Identity amortization, Sybil resistance, cross-identity resource sharing, shared provisioning, influence allocation, Sybil attacks.
\end{IEEEkeywords}}

\begin{document}
	\maketitle
	\IEEEdisplaynontitleabstractindextext
	\IEEEpeerreviewmaketitle
	
	\section{Introduction}
	\label{sec:intro}
	
	\IEEEPARstart{D}{igital} systems frequently recognize multiple identities without directly establishing whether those identities correspond to independent entities~\cite{douceur2002sybil,levine2006survey}. Many Sybil defenses therefore make each identity costly by requiring a resource before admitting it: computational work~\cite{dwork1992pricing}, a verified phone number~\cite{thomas2014dialing}, a network address~\cite{torprop109}, or an attested device~\cite{costan2016intel}. Yet a single such resource can often serve several identities: one phone number can verify several accounts, one address can host several servers, and one device can be enrolled under more than one identity. The defense then charges less per identity than it intends. We call this \emph{identity amortization}: shared provisioning strictly reduces the minimum cost of supporting a set of identities relative to an explicitly defined independent-provisioning baseline. It is not tied to one system; it can arise wherever identities are admitted against a resource that can be shared.
	
	It is also observed in deployed defenses. In Google's phone-verified accounts, Thomas et al.~\cite{thomas2014dialing} found that phone numbers supporting five or more accounts held 58\% of abusive accounts, and public SMS gateways are used primarily to evade verification requirements~\cite{reaves2016sms}. Tor limits the number of relays per network address to raise the cost of running many relays~\cite{torprop109}, and in 2023 it raised that limit. Sections~\ref{sec:experiments} and~\ref{sec:tor} quantify what sharing does to the cost of an attack in each case.
	
	The concept separates three questions: how many distinct identities a system recognizes, what resources and costs are required to support them, and which of those costs must be incurred separately for each identity and which can be shared. These questions are related but not interchangeable. Multiple identities may share a resource without reducing every component of their cost, and a reduction in provisioning cost does not necessarily increase their aggregate influence.
	The cost structure itself is not new: economics calls it \emph{economies of scope}~\cite{panzar1981scope,teece1980scope}, the saving from producing several outputs with a shared input. What a security mechanism needs to know, and what economies of scope does not say, is whether the cost the mechanism intends to impose on each identity is in fact paid for each identity. The term echoes amortized analysis~\cite{tarjan1985amortized}, where a cost is spread over a sequence of operations in time; here it is spread across identities that exist together. From the system's view, a hundred accepted identities are a hundred participants. Under identity amortization they may rest on the resources of a single entity, so the resource requirement, which the system treats as evidence that identities are independent, can be satisfied without that independence. What is shared is not only a cost but the identities the system counts, and influence concentrates without becoming visible.
	
	The question matters for security because a Sybil defense that prices identities assumes the price is paid per identity. If it can be paid once for several, the defense can overestimate the cost of a Sybil threshold, identity counts can overstate independent participation, and raising the price of a shareable resource that honest participants buy can be regressive. Whether cheaper identities yield more influence depends on the mechanism, so these consequences are read from the cost of an attack threshold~\cite{margolin2008quantifying}.
	
	The central question of this paper is therefore: when is a resource requirement that a Sybil defense imposes per identity not paid per identity, how much does this lower the cost of a Sybil attack, and which defense removes the reduction at the least cost to honest participants?
	This paper makes four contributions:
	\begin{itemize}
		\item \emph{Security formalization.} We state the per-identity cost guarantee on which resource-based Sybil defenses rest and define identity amortization as its failure: a strict reduction in minimum provisioning cost relative to a comparable independent baseline, at the resource-instance granularity fixed by the mechanism. Sharing is necessary for identity amortization but not sufficient. Under count-based influence, the failure carries over to the cost of reaching a Sybil threshold; under resource-based influence it need not.
		\item \emph{Decision procedure.} For any catalogue of resource types with instance costs $K_i$, sharing capacities $m_i$, and an identity-specific cost $c$, we give a computable test, with proofs: identity amortization occurs if and only if some fully used resource is cheaper per identity than the cheapest single-identity resource. We determine exactly the number of identities at which it begins, and bound the ratio of independent to shared provisioning cost by an error factor that is attained at scale; the capacity model's bound $(K+c)/(K/m+c)\le m$ is its single-type case and applies to a threshold taken together; one further identity in spare capacity can cost only the identity-specific part (Proposition~\ref{prop:marginal}).
		\item \emph{Defensive interventions.} We compare raising identity-specific costs, raising resource prices, and capping sharing against an adversary who re-optimizes, locate exactly when raising a resource price is regressive, and show that a sharing cap below the onset removes the reduction without changing the price paid by a participant with one identity.
		\item \emph{Evaluation on published and measured data.} On a catalogue built from published phone-verification prices, a per-account estimate overstates the adversary's provisioning cost by a factor between $1.6$ and $4.3$, the onset, safe cap, and interventions follow the propositions, and the cheapest reported market price lies below the least cost of verifying each account with its own number. A case study of Tor's limit on relays per IPv4 address, using the measured network size, shows that raising the limit from two to eight lowered the model cost of reaching a quarter of the hidden-service directories by between $8\%$ and $72\%$, depending on the ratio of address cost to relay cost. These values lie within half a percentage point of the large-population limit of Proposition~\ref{prop:cap-limit}.
	\end{itemize}
	
	Section~\ref{sec:related} reviews related work. Section~\ref{sec:model} defines the model and identity amortization, Section~\ref{sec:capacity} analyzes its existence, onset, and magnitude, and Section~\ref{sec:influence} relates provisioning cost to influence. Section~\ref{sec:security} derives the security consequences. Section~\ref{sec:experiments} evaluates the results on phone verification and Section~\ref{sec:tor} on Tor. Section~\ref{sec:discussion} discusses implications and limitations, and Section~\ref{sec:conclusion} concludes. Proofs are in Appendix~\ref{app:proofs}.

	\section{Related Work}
	\label{sec:related}
	
	This section reviews prior work on Sybil attacks and their defenses and
	on resource-based identity costs and the economics of fake identities,
	and then positions this paper relative to that work.
	
	\subsection{Sybil Attacks and Resistance}
	\label{ssec:rw-sybil}
	
	Douceur~\cite{douceur2002sybil} introduced the Sybil attack and showed
	that, without a logically centralized authority, a distributed system
	cannot in general prevent one entity from presenting multiple
	identities. Subsequent surveys classify the resulting
	defenses~\cite{levine2006survey}. Social-network defenses such as
	SybilGuard, SybilLimit, and
	SybilInfer~\cite{yu2006sybilguard,yu2008sybillimit,danezis2009sybilinfer}
	exploit the limited number of trust relations between Sybil and honest
	regions of a social graph. Alvisi et
	al.~\cite{alvisi2013sok} survey the later development of these defenses
	and their connection to random walks on the trust graph. Proof-of-personhood approaches bind identities
	to physically present persons, for example through pseudonym
	parties~\cite{ford2008offline,borge2017pop}; Siddarth et
	al.~\cite{siddarth2020watchmen} review such protocols. These defenses
	restrict who can obtain an identity. This paper addresses a
	complementary question: when a system admits identities against a
	resource requirement, what an additional identity actually costs.
	
	\subsection{Resource-Based Identity Costs and Their Economics}
	\label{ssec:rw-resource}
	
	A long line of work makes identities or actions costly through
	computation, including pricing via processing~\cite{dwork1992pricing},
	Hashcash~\cite{back2002hashcash}, resource
	testing~\cite{aspnes2005exposing}, and computational puzzles as Sybil
	defenses~\cite{borisov2006puzzles}. Laurie and
	Clayton~\cite{laurie2004proof} concluded that proof-of-work would not stop
	spam, because spammers able to draw on large numbers of compromised
	machines could only be deterred by a cost that would also burden
	legitimate senders. Gupta et al.\ give resource-burning defenses in which,
	under significant attack, honest identities spend asymptotically less
	than the adversary~\cite{gupta2019puzzling,gupta2021bankrupting}. The
	resource in these defenses is consumed, so further participation pays it
	again. A shareable resource, such as a phone number or a network address,
	can instead be held once and used for several identities.
	Trusted hardware~\cite{costan2016intel} can bind an identity to an
	attested device, which fixes the unit that a mechanism counts at the
	device level. Phone
	verification is a widely deployed resource requirement: Thomas et
	al.~\cite{thomas2014dialing} found phone numbers reused across many
	abusive accounts, and Reaves et al.~\cite{reaves2016sms} found that public SMS
	gateways are used primarily to evade account-creation policies. These
	works study resource requirements per identity or document their
	circumvention in a specific setting. They do not ask, in a form that
	applies across mechanisms, when a requirement imposed per identity is
	in fact paid per identity.
	
	Margolin and Levine~\cite{margolin2008quantifying} proposed quantifying
	a system's resistance to the Sybil attack in economic terms. Measurement
	studies document the market for fraudulent accounts: Thomas et
	al.~\cite{thomas2013trafficking} studied the underground market for
	Twitter accounts, and Motoyama et al.~\cite{motoyama2010captcha} analyzed
	CAPTCHA-solving services as an economic activity. These studies show
	that identities can be obtained at low unit cost at scale; they do not
	model which part of that cost is shared across identities.
	\label{ssec:rw-economics}
	
	\subsection{Economies of Scope and Amortized Analysis}
	\label{ssec:rw-scope}
	
	The cost structure behind identity amortization is known in economics
	as \emph{economies of scope}~\cite{panzar1981scope,teece1980scope}:
	providing several outputs jointly can cost less than providing them
	separately when an input is shared. In that setting the joint producer
	is a single, known firm. In a Sybil defense the outputs are identities
	that the mechanism counts as separate participants, so the same saving
	is invisible to the mechanism and lets one entity obtain the influence
	of many. Economies of scope also do not say when the saving begins or
	how large the resulting error is. In the model of this paper, the
	possibility of sharing does not suffice for amortization, a cheaper
	single-identity option delays its onset, and raising a resource price
	can increase the error of a per-identity estimate
	(Sections~\ref{ssec:catalogue} and~\ref{ssec:regressive}). The term
	\emph{amortization} is also used in algorithm
	analysis~\cite{tarjan1985amortized}, where a cost is spread over a
	sequence of operations in time. Here the cost is spread over
	identities that exist together, and the comparison is static.
	
	\subsection{Positioning}
	\label{ssec:rw-positioning}
	
	It is established that a single entity can operate many identities, that
	resource requirements are the main economic defense against this, and
	that such requirements can be reused in practice; the term
	\emph{amortization} itself appears in measurement work on phone-verified
	accounts~\cite{thomas2014dialing}. A related question
	in mechanism design is false-name-proofness, whether an agent gains by
	participating under more than one
	identifier~\cite{yokoo2004falsename,conitzer2010falsename}. Those
	mechanisms are designed so that extra identifiers do not help, often
	while obtaining an identifier is treated as cheap. This paper takes the
	other interface: the admission rule is fixed, and the question is whether
	the resource cost that is supposed to make an identity expensive is paid
	once for many identities. Its contribution is a security
	definition, the failure of the per-identity cost guarantee measured
	against a baseline that keeps the same resource types and prices, and a
	decision procedure built on it for the defender. It does not claim to
	have discovered the phenomenon, or that it occurs in every system.
	
	\section{System Model and Definition}
	\label{sec:model}
	\label{sec:background}
	
	\subsection{Identities, Resources, and Influence}
	\label{ssec:identity-entity}
	
	\begin{definition}[Identity, entity, resource]
		\label{def:identity}
		An \emph{identity} is a representation that a system recognizes and
		treats as a distinct participant, such as an account, a public-key
		identifier, or a node identifier. An \emph{entity} is a controlling
		party that may be represented by one or more identities. A
		\emph{resource} is an asset, capability, or input used to establish,
		maintain, or operate an identity, or to determine its influence.
	\end{definition}
	
	A system may be unable to tell whether identities share a controller, so
	identity count and entity count are different quantities.
	
	\begin{definition}[Influence and sharing]
		\label{def:influence}
		\label{def:sharing}
		For a mechanism $\M$ and a set $S$ of identities, $I_{\M}(S)$ denotes
		the influence that $\M$ assigns to $S$. A resource is \emph{shareable
			across identities} if one instance of it can support more than one
		identity without a separately provisioned instance for each.
	\end{definition}
	
	Influence need not be additive or proportional to the number of
	identities (Section~\ref{sec:influence}), and sharing does not by itself
	lower total cost.
	
	\subsection{Provisioning Strategies and the Independent Baseline}
	\label{ssec:strategies}
	\label{ssec:baseline}
	
	Let $S$ be a finite, nonempty set of identities and $g=|S|$. A
	\emph{resource instance} is the smallest unit of a resource that is
	acquired, held, and assigned as a whole, such as one phone number, one
	device, or one network address. When a single provisioning unit, such
	as a contract or a package, itself supports several identities, that
	unit is one instance. Its granularity is fixed by the mechanism being
	modeled, as the unit the mechanism verifies or the unit in which the
	resource is acquired; it is not chosen so as to produce or avoid
	amortization, and every statement about identity amortization is
	relative to it. A \emph{provisioning strategy} $\pi$
	specifies the instances that are acquired and which identities each
	instance $r$ supports, $\sigma_\pi(r)\subseteq S$; it is feasible if it
	meets the mechanism's requirements for every identity in $S$. A single
	cost function assigns each strategy a cost $C(\pi)\in[0,\infty)$ in one
	unit; it may be nonlinear, for example with bulk discounts.
	
	Let $\Pishared(S)$ be the set of feasible strategies when an instance may
	support more than one identity. The independent baseline keeps the same
	requirements, resource types, and cost function, with one restriction:
	\[
	\begin{aligned}
		\Piind(S)
		=
		\bigl\{\pi\in\Pishared(S):\ &|\sigma_\pi(r)|\le 1\\
		&\text{for all $r$ used by $\pi$}\bigr\}\subseteq\Pishared(S).
	\end{aligned}
	\]
	The restriction limits the identities per instance, not the number of
	instances, so a discount for buying many single-identity instances is
	available in both sets. Any saving that arises from one instance
	supporting several identities is attributable to sharing; any other
	saving is available under both arrangements and does not, by itself,
	constitute identity amortization.
	
	\subsection{Optimal Costs}
	\label{ssec:min-costs}
	\label{ssec:assumptions}
	
	\begin{definition}[Optimal Provisioning Cost]
		\label{def:opt-cost}
		For $X\in\{\mathrm{shared},\mathrm{ind}\}$,
		$C_X(S)=\inf_{\pi\in\Pi_X(S)} C(\pi)$, with $\inf\varnothing=+\infty$.
	\end{definition}
	
	Infima are used because optimal strategies need not exist. One cost
	function, in one unit, is applied to both strategy sets, and the
	baseline differs from shared provisioning only by the restriction above.
	
	\begin{lemma}[Cost ordering]
		\label{prop:cost-order}
		For every finite, nonempty identity set $S$,
		$\Cshared(S)\le\Cind(S)$.
	\end{lemma}
	
	The lemma follows from $\Piind(S)\subseteq\Pishared(S)$ alone; it does
	not require linear costs or the existence of optimal strategies. It
	states that permitting sharing never raises the optimal cost, not that
	it lowers it.
	
	\subsection{Definition of Identity Amortization}
	\label{ssec:def}
	
	Resource-based Sybil defenses rest on a cost assumption: requiring a
	resource for each identity is meant to make supporting several
	identities cost at least as much as supporting each one with its own
	resources.
	
	\begin{definition}[Per-Identity Cost Guarantee]
		\label{def:guarantee}
		For a finite, nonempty identity set $S$ with $\Cind(S)<+\infty$,
		the \emph{per-identity cost guarantee} holds for $S$ if
		$\Cshared(S)\ge\Cind(S)$.
	\end{definition}
	
	By Lemma~\ref{prop:cost-order}, the guarantee holds exactly when
	$\Cshared(S)=\Cind(S)$. It concerns provisioning cost only; that it holds
	does not by itself make a system resistant to Sybil attacks.
	
	\begin{definition}[Identity Amortization]
		\label{def:ia}
		Let $S$ be a finite, nonempty identity set with
		$0<\Cind(S)<+\infty$. The \emph{amortization ratio} of $S$ is
		$\alpha(S)=\Cshared(S)/\Cind(S)\in[0,1]$, and \emph{identity
			amortization occurs for $S$} if $\alpha(S)<1$, equivalently
		$\Cshared(S)<\Cind(S)$.
	\end{definition}
	
	Identity amortization occurs for $S$ if and only if the per-identity
	cost guarantee fails for $S$. It is the strict cost reduction that may
	result from sharing, not the possibility of sharing. Sharing is
	necessary for it, since any strategy cheaper than $\Cind(S)$ lies outside
	$\Piind(S)$ and so lets some instance support two or more identities; it
	is not sufficient, since a mechanism may permit sharing without any
	shared strategy costing less than $\Cind(S)$. Because both costs refer
	to the same identity set, $\alpha(S)$ is also the ratio of the average
	optimal costs per identity, $\Cshared(S)/|S|$ and $\Cind(S)/|S|$:
	identity amortization means a strictly lower average cost per identity
	when resources may be shared.
	
	\paragraph{Relation to the Sybil attack} A Sybil attack exploits
	identities that are cheap to create, so that one entity can present
	many of them. A defense that requires a resource for each identity
	addresses this by making identities costly. Identity amortization is a
	second route to the same gap: the requirement may be costly for each
	identity in isolation, yet the adversary does not pay it in full for
	each identity. Whether this lowers the cost of an attack depends on how
	the mechanism assigns influence. When every provisioned identity is
	accepted and influence is assigned per accepted identity, the cost of
	reaching a Sybil threshold equals the cost of provisioning the required
	number of identities (Section~\ref{ssec:count-based}); a failure of the
	guarantee at that number therefore makes the threshold cheaper than the
	independent estimate, by the factor $R_\tau$ of
	Section~\ref{ssec:threshold}. This factor is a ratio of costs; it is not
	a multiplier on the adversary's influence or on the probability that an
	attack succeeds. When influence is assigned per unit of an
	influence-bearing resource, cheaper identities need not yield more
	influence, although sharing can still lower the cost of influence
	(Section~\ref{ssec:resource-based}).
	
	\paragraph{Scope}
	\label{ssec:scope}
	The definition is relative to the baseline, the cost function, and the
	granularity of resource instances, and it concerns provisioning cost
	only. It does not imply that every shareable resource produces
	amortization, that the marginal cost of an identity is zero, that more
	identities yield more influence, or that amortization leads to a
	successful Sybil attack. How large the reduction can be, when it lowers
	the cost of attaining influence, and when it lowers the cost of a Sybil
	threshold are examined in
	Sections~\ref{sec:capacity}--\ref{sec:security}.
	The comparison is static. $C(\pi)$ is the total cost of provisioning
	every identity in $S$ so that the mechanism accepts it, and an instance
	supports, at once, the identities assigned to it. How provisioning costs
	evolve over time, including the sequential reuse of an instance, is
	outside the model.
	
	\section{Identity Amortization and the Cost of Multi-Identity Provisioning}
	\label{sec:capacity}
	
	\subsection{Count-Level Costs}
	\label{ssec:provisioning-costs}
	
	An adversary may choose which identities to provision, so the cost of
	$g$ identities is derived from the set-level costs.
	
	\begin{definition}[Count-Level Provisioning Cost]
		\label{def:count-cost}
		For $X\in\{\mathrm{shared},\mathrm{ind}\}$ and an integer $g\ge1$,
		$\bar C_X(g)=\inf\{C_X(S):\ S\ \text{finite},\ |S|=g\}$.
	\end{definition}
	
	By Lemma~\ref{prop:cost-order},
	\begin{equation}
		\Cbshared(g)\le\Cbind(g)\qquad\text{for every } g\ge1,
		\label{eq:cost-ordering}
	\end{equation}
	a weak ordering only. It is strict only when some shared arrangement
	for $g$ identities costs less than every independent one; the existence
	of a shareable resource does not guarantee this, since costs that
	remain identity-specific, or a cheaper single-identity alternative, may
	offset the benefit of sharing. The shared and independent infima may be
	attained on different identity sets. For $g$ with $0<\Cbind(g)<+\infty$ we write
	$\Delta(g)=\Cbind(g)-\Cbshared(g)$ for the absolute saving and
	$\alpha(g)=\Cbshared(g)/\Cbind(g)$ for the count-level amortization
	ratio. The per-identity cost guarantee holds at $g$ if
	$\Cbshared(g)\ge\Cbind(g)$, and identity amortization occurs at $g$ if
	$\alpha(g)<1$.
	
	\subsection{A Parametric Capacity Model}
	\label{ssec:capacity-model}
	\label{ssec:capacity-and-specific-costs}
	
	In the capacity model, every identity needs one unit of capacity of a
	single shareable resource type and incurs a non-shareable
	identity-specific cost $c\ge0$; every instance costs $K>0$, whether or not
	its capacity is used, and supports at most $m\ge1$ identities; the cost of
	a strategy is $K$ times the number of instances plus $c$ times the number
	of identities; and any assignment that respects the capacity is feasible.
	The same instance type is available under both arrangements.
	
	\begin{lemma}[Costs in the capacity model]
		\label{lem:capacity-costs}
		For every integer $g\ge1$,
		\begin{equation}
			\Cbind(g)=g(K+c),\qquad
			\Cbshared(g)=K\left\lceil\frac{g}{m}\right\rceil+gc,
			\label{eq:shared-capacity-model}
		\end{equation}
		and both infima are attained.
	\end{lemma}
	
	Write $\alpha_m=(K/m+c)/(K+c)$.
	
	\begin{proposition}[Cost reduction in the capacity model]
		\label{prop:capacity-cost-reduction}
		For every integer $g\ge1$:
		\begin{enumerate}[label=(\alph*),leftmargin=*]
			\item $\Delta(g)=K(g-\lceil g/m\rceil)$, which does not depend on
			$c$ and is strictly positive if and only if $m>1$ and $g>1$.
			\item $\alpha(g)\ge\alpha_m$, with equality if and only if $m$ divides
			$g$, and $\lim_{g\to\infty}\alpha(g)=\alpha_m$.
		\end{enumerate}
	\end{proposition}
	
	Proofs of this and the following results are in
	Appendix~\ref{app:proofs}. The proposition separates the roles of the
	parameters. If $m=1$, sharing gives no capacity advantage, and if
	$g=1$, both arrangements need one instance. A larger capacity $m$
	cannot raise the shared cost, because $\lceil g/m\rceil$ is
	nonincreasing in $m$. The identity-specific cost $c$ does not change
	the absolute saving but bounds the relative one: sharing cannot remove
	more than the fraction $\frac{K}{K+c}(1-\frac1m)$ of the independent
	cost. The same holds for any cost component that every strategy pays in
	the same amount under both arrangements.
	
	The saving also grows in steps. Under shared provisioning, a new
	instance is needed only when the capacity of the existing ones is
	exhausted, so an additional identity costs $c$ between capacity
	boundaries and $K+c$ at a boundary, whereas under independent
	provisioning it always costs $K+c$. Hence $\Delta(g+1)-\Delta(g)$ is $K$
	if $m\nmid g$ and $0$ if $m\mid g$, and $\Delta(g)/g\to K(1-1/m)$.
	
	\begin{proposition}[Marginal overestimation]
		\label{prop:marginal}
		Let $c>0$ and let $g\ge1$ be an integer. One further identity costs
		$c$ under shared provisioning if $m\nmid g$, and $K+c$ if $m\mid g$.
		Against the independent cost $K+c$ of that identity, the ratio is
		$(K+c)/c$ in the first case and $1$ in the second. The ratio
		$(K+c)/c$ is greater than $m$ if and only if $K>(m-1)c$.
	\end{proposition}
	
	Proposition~\ref{prop:marginal} separates this step from the cost of a
	whole set. The ratio for the set tends to $(K+c)/(K/m+c)$, which is at
	most $m$; the ratio for one further identity in spare capacity is
	$(K+c)/c$, which need not be.
	
	\begin{proposition}[Large-population cap increase]
		\label{prop:cap-limit}
		Let $k$ and $\ell$ be integers with $\ell>k\ge1$. As $g\to\infty$,
		the relative reduction in shared cost from raising the capacity
		from $k$ to $\ell$ tends to $(\ell-k)/\ell$ if $c=0$. If $c>0$ and
		$r=K/c$, it tends to
		\begin{equation}
			\frac{r(\ell-k)}{\ell(r+k)},
			\label{eq:cap-limit}
		\end{equation}
		which is strictly increasing in $r$ and lies strictly between $0$
		and $(\ell-k)/\ell$.
	\end{proposition}
	
	\emph{Proof idea.} Dividing both shared costs by $g$ and using $\lceil g/t\rceil/g\to1/t$, their ratio tends to $(K/\ell+c)/(K/k+c)$; one minus this ratio is the stated limit.
	
	A rise from capacity $2$ to capacity $8$ therefore tends to
	$3r/(4(r+2))$, which equals $1/2$ at $r=4$. The limit depends on the
	price ratio and the two capacities, not on the number of identities or
	on a target share. Section~\ref{sec:tor} compares it with the reduction
	at a measured network size.
	
	The model gives the independent baseline the same instance type as the
	shared arrangement. When a cheaper resource suffices for a single
	identity, the cheapest independent arrangement may use it, and a strict
	reduction then occurs only above a threshold number of identities.
	Section~\ref{ssec:catalogue} determines this threshold exactly.
	
	\subsection{Resource Catalogues: Existence, Onset, and Magnitude}
	\label{ssec:catalogue}
	
	The capacity model fixes a single resource type. A defender who
	assesses a deployed requirement typically faces several ways of
	satisfying it, such as resources of different sizes, prices, and
	capacities, and an adversary chooses whichever combination is cheapest.
	A \emph{resource catalogue} is a finite set of types
	$i=1,\dots,n$, where one instance of type $i$ costs $K_i>0$ and supports
	at most $m_i\ge1$ identities, every identity incurs the
	identity-specific cost $c\ge0$, instances are indivisible and charged in
	full, and any combination may be acquired. Under shared provisioning,
	$x_i$ instances of each type support $g$ identities exactly when
	$\sum_ix_im_i\ge g$; under independent provisioning, each instance
	supports at most one identity. With
	\[
	K_{\min}=\min_iK_i
	\qquad\text{and}\qquad
	\rho^*=\min_i\frac{K_i}{m_i}\ \le\ K_{\min},
	\]
	the cheapest single-identity instance and the cheapest instance cost
	per supported identity, the count-level costs are
	\[
	\Cbind(g)=g(K_{\min}+c),
	\]
	\[
	\begin{aligned}
		\Cbshared(g)=\min\Bigl\{\textstyle\sum_ix_iK_i:\ &x\in\mathbb{N}^n,\\
		&\textstyle\sum_ix_im_i\ge g\Bigr\}+gc.
	\end{aligned}
	\]
	The capacity model is the case $n=1$. We call type $i$
	\emph{amortizing} if $K_i<m_iK_{\min}$, that is, if a fully used instance
	of type $i$ is cheaper per identity than the cheapest single-identity
	instance.
	
	\begin{lemma}[Per-identity floor]
		\label{lem:floor}
		The instance cost of every combination that supports $g$ identities
		is at least $g\rho^*$.
	\end{lemma}
	
	\begin{proof}
		$\sum_ix_iK_i=\sum_i(x_im_i)(K_i/m_i)\ge\rho^*\sum_ix_im_i\ge\rho^*g$.
	\end{proof}
	
	\begin{proposition}[Existence and onset]
		\label{prop:onset}
		Identity amortization occurs for some number of identities if and
		only if $\rho^*<K_{\min}$, equivalently if some type is amortizing.
		In that case, writing
		\[
		g^*=\min_{i\ \text{amortizing}}
		\left(\left\lfloor\frac{K_i}{K_{\min}}\right\rfloor+1\right),
		\]
		$\Cbshared(g)<\Cbind(g)$ if and only if $g\ge g^*$.
	\end{proposition}
	
	\emph{Proof idea.} If $\rho^*=K_{\min}$, the per-identity floor (Lemma~\ref{lem:floor}) rules out any saving. If type $i$ is amortizing, one of its instances supports $\lfloor K_i/K_{\min}\rfloor+1$ identities for less than that many cheapest single-identity instances, and adding cheapest instances keeps the saving at every larger count; below $g^*$, any instance of an amortizing type alone costs at least $gK_{\min}$, and every other instance costs at least $K_{\min}$ per identity it supports.
	
	The test uses only prices and capacities: a defender can determine
	whether the per-identity cost guarantee can fail, and from which number
	of identities onward, without modeling the adversary further. The onset
	depends on the cheapest single-identity option as well as on the
	shareable resource: a cheaper single-identity resource lowers $K_{\min}$
	and so delays the onset. In the capacity model, $g^*=2$, as in
	Proposition~\ref{prop:capacity-cost-reduction}(a), and the reduction
	remains strict at every larger count.
	
	A cap on sharing, which limits every instance to at most $k$
	identities, replaces each $m_i$ by $\min(m_i,k)$. Its effect follows
	from Proposition~\ref{prop:onset}.
	
	\begin{corollary}[Safe sharing cap]
		\label{cor:safe-cap}
		If identity amortization occurs in a catalogue, then under a cap $k$
		it occurs for no number of identities if and only if $k\le g^*-1$.
	\end{corollary}
	
	\emph{Proof idea.} A cap $k$ leaves $K_{\min}$ unchanged, and type $i$ remains amortizing exactly when $kK_{\min}>K_i$; no type does so exactly when $k\le g^*-1$.
	
	\begin{proposition}[Magnitude]
		\label{prop:error-bound}
		For every $g\ge1$,
		\[
		\frac{\Cbind(g)}{\Cbshared(g)}\ \le\ E:=\frac{K_{\min}+c}{\rho^*+c}.
		\]
		The bound is attained whenever $g$ is a multiple of $m_b$ for a type
		$b$ with $K_b/m_b=\rho^*$, and the ratio tends to $E$ as $g\to\infty$.
		Moreover, $E\le m_b$ for every such type $b$.
	\end{proposition}
	
	\emph{Proof idea.} The floor gives $\Cbshared(g)\ge g(\rho^*+c)$, while $\Cbind(g)=g(K_{\min}+c)$; $g/m_b$ instances of a type with $K_b/m_b=\rho^*$ attain the floor.
	
	The factor $E$ is the largest error made by an estimate that prices
	every identity at the cheapest single-identity option, and it is reached
	at scale. For a single type it is $1/\alpha_m=(K+c)/(K/m+c)$. Whether
	the guarantee can fail, from which count onward, and by how much are
	thus determined by prices and capacities alone. These results concern
	provisioning cost; whether lower provisioning cost yields more
	influence depends on the mechanism, which Section~\ref{sec:influence}
	examines.
	
	\section{From Identity Count to Influence}
	\label{sec:influence}
	
	A failure of the per-identity cost guarantee becomes security-relevant
	when it lowers the cost of attaining a given level of influence, and
	that depends on the mechanism. In count-based mechanisms, additional
	identities may directly increase an entity's influence. In
	resource-based mechanisms, influence is determined by the resources an
	entity controls, while shared infrastructure can lower the cost of
	acquiring or operating them.
	
	\subsection{Cost of Attaining a Target Influence}
	\label{ssec:target-influence}
	
	For $X\in\{\mathrm{ind},\mathrm{shared}\}$, let $\Pi_X$ be the union of
	the strategy sets $\Pi_X(S)$ over all finite, nonempty $S$, so that each
	strategy $\pi$ determines its identity set $S(\pi)$, and let
	$I_{\M}(\pi)$ be the influence its configuration produces under
	mechanism $\M$. The minimum cost of attaining influence $v$ is
	\[
	\mathcal{K}^{\M}_X(v)
	=\inf\bigl\{C(\pi):\ \pi\in\Pi_X,\ I_{\M}(\pi)\ge v\bigr\},
	\]
	and $\mathcal{K}^{\M}_{\mathrm{shared}}(v)\le\mathcal{K}^{\M}_{\mathrm{ind}}(v)$
	because $\Pi_{\mathrm{ind}}\subseteq\Pi_{\mathrm{shared}}$. Where the
	independent cost is positive and finite, we write
	$\alpha_{\M}(v)=\mathcal{K}^{\M}_{\mathrm{shared}}(v)/\mathcal{K}^{\M}_{\mathrm{ind}}(v)$.
	This ratio fixes a mechanism-level outcome rather than a number of
	identities, so it accommodates mechanisms in which influence depends on
	identity count, on resource allocation, or on both, and the
	cost-minimizing strategies under the two arrangements may differ.
	
	\subsection{Count-Based Influence}
	\label{ssec:count-based}
	
	In a count-based mechanism, every provisioned identity is accepted and
	contributes $w>0$ units of influence, so reaching $v$ requires
	$q_{\M}(v)=\lceil v/w\rceil$ identities. When the minimum provisioning
	cost is nondecreasing in the number of identities, as it is in the
	models of Section~\ref{sec:capacity},
	\[
	\begin{gathered}
		\mathcal{K}^{\M}_{\mathrm{ind}}(v)=\Cbind\bigl(q_{\M}(v)\bigr),\\
		\mathcal{K}^{\M}_{\mathrm{shared}}(v)=\Cbshared\bigl(q_{\M}(v)\bigr),
	\end{gathered}
	\]
	so $\alpha_{\M}(v)=\alpha(q_{\M}(v))$. Identity amortization then lowers
	the cost of attaining the target exactly when it lowers the cost of the
	required identities, and for a fixed budget it buys more influence. A mechanism that counts identities
	as independent participants thus assigns additional influence without
	requiring a corresponding increase in independently provisioned
	resources. A phone number that a system permits to support several
	accounts illustrates this: the number's cost is shared, while
	identity-specific costs and acceptance rules remain separate.
	
	\subsection{Resource-Based Influence}
	\label{ssec:resource-based}
	
	In a resource-based mechanism, influence depends on an
	influence-bearing resource, such as computational power, and splitting
	a fixed quantity of it among more identities does not by itself raise
	influence. Identity amortization still occurs, in two places. First,
	the identities through which the resource is deployed may share
	provisioning costs: if the cost $R_{\M}(v)$ of the resource needed for
	$v$ is independent of how identities are provisioned, and the mechanism
	requires $q$ identities to deploy it, the cost of $v$ is
	$R_{\M}(v)+\bar C_X(q)$, and sharing lowers it by
	$\Cbind(q)-\Cbshared(q)$. Second, shared facilities and operations can
	lower the cost of acquiring or operating the influence-bearing resource
	itself, so that $R_{\M}(v)$ is lower for a participant who shares them
	than for one who provisions them alone. Neither reduction yields more
	influence per unit of resource, so the bridge to a Sybil threshold of
	Section~\ref{ssec:threshold} does not apply directly. Both, however,
	lower the cost of a given level of influence for participants who share
	infrastructure, and so favor concentration of influence among them.
	
	\section{The Security Consequences of Identity Amortization}
	\label{sec:security}
	
	This section derives the consequences of a failure of the per-identity
	cost guarantee for the cost of a Sybil threshold, for identity counts as
	indicators of independent participation, and for the design of
	defenses.
	
	\subsection{Cheaper Sybil Thresholds and Overestimated Attack Costs}
	\label{ssec:threshold}
	
	Consider a count-based mechanism in which each accepted identity
	receives one vote and every provisioned identity is accepted. With $N$
	honest identities and $g$ adversarial ones, the adversary's share is
	$g/(N+g)$, and reaching a share $\tau\in(0,1)$ requires
	\[
	g_\tau=\left\lceil\frac{\tau N}{1-\tau}\right\rceil
	\]
	identities. This translates an influence objective into a provisioning
	requirement, at minimum costs $\Cbind(g_\tau)$ and $\Cbshared(g_\tau)$
	(Section~\ref{ssec:count-based}). Whenever both are positive and finite,
	the \emph{attack-cost overestimation factor} is
	\[
	R_\tau=\frac{\Cbind(g_\tau)}{\Cbshared(g_\tau)}=\frac{1}{\alpha(g_\tau)}.
	\]
	It measures how much larger the estimated cost is when independent
	provisioning is assumed; it is a cost ratio, not a multiplier on the
	probability of attack success. $R_\tau>1$ exactly when the per-identity
	cost guarantee fails at $g_\tau$.
	
	\begin{proposition}[Bound on attack-cost overestimation]
		\label{prop:overestimation}
		In the capacity model of Section~\ref{ssec:capacity-model},
		\[
		R_\tau=\frac{g_\tau(K+c)}{K\lceil g_\tau/m\rceil+g_\tau c}
		\ \le\ \frac{K+c}{K/m+c}\ \le\ m.
		\]
		The first inequality is an equality if and only if $m$ divides
		$g_\tau$, and the second if and only if $m=1$ or $c=0$.
	\end{proposition}
	
	\emph{Proof idea.} Replacing $\lceil g_\tau/m\rceil$ by $g_\tau/m$ gives the first bound, and $K/m+c\ge(K+c)/m$ gives the second.
	
	A defense calibrated to an assumed cost per identity may therefore
	present a barrier up to $(K+c)/(K/m+c)$ times lower than intended, and up
	to $m$ times lower when identity-specific costs are negligible. For any
	resource catalogue, the costs are nondecreasing in $g$, so $R_\tau\le E$
	by Proposition~\ref{prop:error-bound}. Admission controls, detection, and
	other attack costs lie outside this comparison.
	
	\subsection{Scale and Apparent Participation}
	\label{ssec:scale}
	\label{ssec:decentralization}
	
	Once a shareable instance has been provisioned, further identities cost
	only $c$ each until its capacity is exhausted, whereas an independently
	provisioned identity costs $K+c$. Infrastructure already paid for is
	therefore latent identity capacity, and an entity that holds it can
	expand its identity set more cheaply than one that must provision each
	identity independently; the cost of a new instance recurs only at
	capacity boundaries (Section~\ref{ssec:capacity-model}). For an
	adversary who already holds spare capacity, the nominal price overstates
	each further identity that fits in it by $(K+c)/c$ when $c>0$, and this
	factor exceeds $m$ whenever $K>(m-1)c$
	(Proposition~\ref{prop:marginal}). Where
	additional accepted identities translate into influence, this lowers
	the cost of scaling adversarial participation.
	
	For the same reason, a population of visible identities can be cheaper
	to produce than an equally large population of independently
	provisioned participants: in the capacity model, by $\Delta(q)$ for $q$
	identities. A large number of visible participants therefore need not
	reflect an equally large number of independently provisioned resource
	bases. Identity counts, and any metric computed over declared
	identities, including resource shares attributed to declared identities
	rather than to the entities that provision them, can overstate the
	independence of the participation they represent.
	
	\subsection{Regressive Identity Pricing and Defense Design}
	\label{ssec:regressive}
	
	The same nominal requirement can impose different effective costs on
	different actors. A participant who establishes one identity without
	existing infrastructure pays the full instance cost $K$ together with
	$c$; an entity that has already provisioned a shareable instance pays
	only $c$ for each further identity while spare capacity remains. In the
	capacity model, the honest participant and the
	adversary buy the same type: an increase in $K$ is borne in full by a
	participant with one identity, but only about one $m$-th per identity by
	an entity that spreads it across many, so the ratio $(K+c)/(K/m+c)$
	rises with $K$ when $m>1$ and $c>0$, and equals $m$ for every $K$ when
	$c=0$; the absolute per-identity gap $K(1-1/m)$ widens with $K$. An
	increase in the identity-specific
	cost $c$, by contrast, is paid for every identity and moves the ratio
	toward one.
	
	With several resource types, the two need not buy the same type, and the
	adversary re-optimizes after any change. Consider an honest participant
	who provisions one identity, at cost $H=K_{\min}+c$, and an adversary who
	provisions $g$ identities at cost $\Cbshared(g)$ in a resource catalogue.
	
	\begin{proposition}[Defensive interventions]
		\label{prop:interventions}
		Let $\delta>0$.
		\begin{enumerate}[label=(\alph*),leftmargin=*]
			\item Raising $c$ by $\delta$ raises $\Cbshared(g)$ by exactly
			$g\delta$ and $H$ by $\delta$; if $\rho^*<K_{\min}$, it strictly
			decreases $E$.
			\item Raising the price $K_j$ of one type by $\delta$ raises
			$\Cbshared(g)$ by at most $\delta x_j$, where $x$ is any combination
			that was optimal before the change. It raises $H$ only if $j$ is the
			unique type attaining $K_{\min}$, and then by
			$\min(\delta,K'-K_j)$, where $K'$ is the cheapest price among the
			other types ($K'=+\infty$ if there is none).
			\item The error factor $E=(K_{\min}+c)/(\rho^*+c)$ increases with
			$K_{\min}$ and decreases with $\rho^*$.
			\item A sharing cap $k\ge1$ leaves $H$ unchanged and, by
			Corollary~\ref{cor:safe-cap}, removes identity amortization for
			every $g$ if and only if $k\le g^*-1$.
		\end{enumerate}
	\end{proposition}
	
	The proposition locates the regressive effect. Raising the price of the
	cheapest single-identity resource, when the adversary's cheapest
	per-identity resource is a different type, raises the honest cost while
	leaving $\rho^*$, and possibly $\Cbshared(g)$, unchanged, so the error
	factor grows. Raising the price of a type that only the adversary uses
	raises its cost by at most $\delta x_j$ at no cost to $H$. A sharing cap
	below the onset removes the reduction without changing the price paid by
	a participant with one identity. These statements concern provisioning
	cost only: enforcing a cap requires linking identities to the instance
	that supports them, and legitimate participants who need several
	identities, such as several accounts on one phone number, bear part of
	its cost; neither effect is modeled here. A defense that relies on the
	cost of identity creation should therefore ask not only how expensive a
	resource is, but how much of that cost must be paid separately for
	every identity, so that the cost of a shared resource is not mistaken
	for a cost that an adversary must pay again for each identity.
	
	\section{Model-Based Evaluation on Phone Verification}
	\label{sec:experiments}
	
	This section instantiates the catalogue results of
	Section~\ref{ssec:catalogue} and the interventions of
	Proposition~\ref{prop:interventions} on phone verification, using only
	prices and observations published by Thomas et
	al.~\cite{thomas2014dialing} for Google phone-verified accounts. The
	minima are computed exactly by dynamic programming and compared with the
	propositions. The calculations are model-based; their inputs are
	published measurements, and costs are in US dollars per account.
	
	\subsection{Catalogue}
	\label{ssec:eval-setup}
	
	Thomas et al.~\cite{thomas2014dialing} are used because, to our
	knowledge, that study is the only peer-reviewed source that reports
	code prices, SIM prices, observed reuse, and identity-specific costs
	as one matched set for a deployed system. The purpose is to test the
	decision procedure on those inputs, not to estimate current costs.
	Thomas et al.\ report two ways of obtaining a verification for an
	account~\cite[\S4.1.2]{thomas2014dialing}: verification codes sold as a
	service, at \$140 per thousand codes, and bulk SIM cards, at \$140--420
	per thousand. A code verifies one account. A number is not limited to
	one account: phone clusters of size five or greater contained $58\%$ of
	abusive accounts~\cite[\S4.2.1]{thomas2014dialing}, so we take five as
	the capacity of a SIM card. This is a lower bound on what was observed,
	not the provider's limit, which the study does not report. The
	identity-specific cost is $c=\$0.0057$: a human-solved CAPTCHA at \$0.70
	per thousand and an email account at \$5 per
	thousand~\cite[\S2.2]{thomas2014dialing}. The catalogue is therefore
	\[
	(K_i,m_i)\in\{(\$0.14,1),\ (K_s,5)\},
	\]
	with $K_s\in\{\$0.14,\$0.42\}$,
	the two ends of the reported SIM price range. Free VOIP numbers and
	proxies are left out, because the study gives no comparable price per
	account for them. The honest participant needs one account, at
	$H=K_{\min}+c=\$0.1457$.
	
	\subsection{Existence, Onset, and Magnitude}
	\label{ssec:eval-catalogue}
	
	\begin{table}[t]
		\caption{Phone-verification catalogue built from published
			prices~\cite{thomas2014dialing}: onset $g^*$, safe sharing cap
			$k_{\mathrm{safe}}$, and error factor $E$ of a per-account estimate.}
		\label{tab:phone}
		\centering
		\begin{tabular}{lrrr}
			\hline
			SIM price $K_s$ & $g^*$ & $k_{\mathrm{safe}}$ & $E$\\
			\hline
			\$0.14 & $2$ & $1$ & $4.32$\\
			\$0.42 & $4$ & $3$ & $1.62$\\
			\hline
		\end{tabular}
	\end{table}
	
	Here $K_{\min}=\$0.14$, and a SIM card is an amortizing type at both
	prices, since $K_s<5K_{\min}$. Table~\ref{tab:phone} gives the results.
	Proposition~\ref{prop:onset} predicts the onset
	$g^*=\lfloor K_s/K_{\min}\rfloor+1$, that is, $2$ and $4$, and the
	dynamic program finds the first strict reduction at exactly these
	counts; below $g^*$ no mix that shares a SIM card is cheaper than one code per account.
	Proposition~\ref{prop:error-bound} gives
	$E=(K_{\min}+c)/(K_s/5+c)$, that is, $4.32$ and $1.62$. Both values are
	attained at $g=60$, a multiple of the capacity: there the shared cost is
	\$2.02 and \$5.38 against an independent cost of \$8.74. A defense that
	prices every account at one number therefore overestimates the
	adversary's provisioning cost by a factor between $1.6$ and $4.3$. By
	Corollary~\ref{cor:safe-cap}, a limit of one account per number removes
	the reduction at the lower SIM price, and a limit of three at the higher.
	
	The published market prices are consistent with this. The cheapest
	phone-verified Google accounts sold for \$85 per
	thousand~\cite[\S3.1.2]{thomas2014dialing}, that is, \$0.085 per account.
	This is below $H=\$0.1457$, the least cost of provisioning one account
	with its own number in this catalogue, and lies between the shared
	floors $\rho^*+c$ of the two catalogues, \$0.034 at $K_s=\$0.14$ and
	\$0.090 at $K_s=\$0.42$. A sale below the independent cost
	is possible only if the cost of a number is spread over several accounts,
	or if a cheaper resource is used. The study also reports a market-wide
	price drop of $30$--$40\%$ while free VOIP numbers were heavily used,
	reversed after Google penalized registrations from frequently abused
	VOIP providers~\cite[\S3.1.2, \S4.1.2]{thomas2014dialing}. This, too, is
	consistent with a cheaper shareable resource lowering the price per
	account. 	These observations do not establish the cause of the prices.
	
	\subsection{Interventions}
	\label{ssec:eval-interventions}
	
	\begin{table}[t]
		\caption{Interventions on the phone catalogue with $K_s=\$0.42$ and
			$g=60$, the adversary re-optimizing after each change: change in
			the adversary's cost $\Delta_A$ and in the honest cost $\Delta H$
			(US dollars), and the error factor $E$.}
		\label{tab:interventions}
		\centering
		\begin{tabular}{lrrr}
			\hline
			Intervention & $\Delta_A$ & $\Delta H$ & $E$\\
			\hline
			None & -- & -- & $1.62$\\
			$c+\$0.01$ & $0.60$ & $0.01$ & $1.56$\\
			code price $+\$0.01$ & $0$ & $0.01$ & $1.74$\\
			SIM price $+\$0.01$ & $0.12$ & $0$ & $1.59$\\
			limit $4$ per number & $1.26$ & $0$ & $1.32$\\
			limit $3$ per number & $3.36$ & $0$ & $1.00$\\
			\hline
		\end{tabular}
	\end{table}
	
	Table~\ref{tab:interventions} applies each intervention once at
	$K_s=\$0.42$, where the honest participant buys a code and the
	adversary buys SIM cards. Each row follows
	Proposition~\ref{prop:interventions}. A higher identity-specific cost
	raises the adversary's cost by exactly $g\delta=\$0.60$ and the honest
	cost by $\delta$. A higher code price raises only the honest cost and the
	error factor: this is the regressive case of
	Section~\ref{ssec:regressive}. A higher SIM price raises the adversary's
	cost by $\delta x_j=\$0.12$, since it uses $12$ cards, at no cost to the
	honest participant. A limit of four accounts per number still leaves
	amortization from four accounts on; a limit of three, the safe cap,
	removes it, and the adversary pays the independent cost of \$8.74.
	Reaching an error factor of $1.2$ instead requires $c\ge\$0.196$
	(Section~\ref{ssec:eval-scope}) and raises the honest cost from
	$\$0.1457$ to $\$0.336$, about $2.3$ times, whereas the limit of three
	leaves that cost unchanged. Raising $c$ charges every participant with
	one account. The limit does not, but it falls in part on someone who
	legitimately keeps several accounts on one number, and that part is
	not modeled here.
	
	\subsection{Sensitivity, Current Prices, and Scope}
	\label{ssec:eval-scope}
	
	The results depend on the published prices, which date from
	2013--2014, on the capacity of five, which is a lower bound on observed
	reuse rather than a stated limit, and on the omission of free VOIP
	numbers and proxies. They illustrate the propositions on measured
	inputs; they are not a measurement of attack costs, and they do not
	establish what a deployed admission rule imposes today. A current
	code price, SIM price, identity-specific cost, and observed reuse are
	not available as one matched set. We therefore keep the 2013--2014
	catalogue as the reference and check it below against a current code
	price, replacing only the input for which a current published value
	exists.
	
	On these prices the onset does not depend on $c$, since $c$ cancels in
	Proposition~\ref{prop:onset}. At $K_s=\$0.42$ a SIM card is amortizing
	only if its capacity $m$ satisfies $K_s<mK_{\min}$, that is,
	$m\ge4$. At capacity $3$ the two resources have the same cost per
	account, the error factor is $1$, and the reduction disappears; at
	capacity $4$ the error factor is $1.32$. At $K_s=\$0.14$ the SIM card
	is amortizing for every capacity of at least $2$, and at capacity $2$
	the error factor is $1.92$. Bringing the error factor at $K_s=\$0.42$
	and capacity $5$ down from its published value to $1.2$ requires
	$c\ge\$0.196$, against the published $\$0.0057$. The reduction is
	therefore stable under a small change in the identity-specific cost.
	It depends on whether a number supports at least four accounts when the
	SIM card is at the high published price.
	
	\begin{table}[t]
		\caption{Phone catalogue with current verification prices
			(COTSI, July 2024--July 2025, average over
			platforms~\cite{dek2025manipulation,cam2025botarmy}). SIM price,
			capacity $5$, and $c=\$0.0057$ are from~\cite{thomas2014dialing}:
			a composite catalogue, not a matched set.}
		\label{tab:phone-current}
		\centering
		\begin{tabular}{llrrrr}
			\hline
			Country & $K_s$ & $K_{\min}$ & $g^*$ & $k_{\mathrm{safe}}$ & $E$\\
			\hline
			US (\$0.26) & \$0.14 & \$0.14 & $2$ & $1$ & $4.32$\\
			US (\$0.26) & \$0.42 & \$0.26 & $2$ & $1$ & $2.96$\\
			UK (\$0.10) & \$0.14 & \$0.10 & $2$ & $1$ & $3.14$\\
			UK (\$0.10) & \$0.42 & \$0.10 & $5$ & $4$ & $1.18$\\
			\hline
		\end{tabular}
	\end{table}
	
	Table~\ref{tab:phone-current} replaces only the code price. Where a
	verification is more expensive, sharing a SIM card pays from two
	accounts on even at the higher SIM price; where it is cheaper, a SIM
	card at \$0.42 pays only from five accounts on. When the SIM card is
	the cheaper resource, it is also the cheapest single-identity option,
	and $K_{\min}=K_s$. The error factor lies between $1.18$ and $4.32$.
	Google states that it limits the number of accounts each phone number
	can create, without publishing the limit~\cite{googleverify}. With these
	prices, removing the reduction would require a limit of one account per
	number in the US, and of one to four in the UK.
	
	\section{Case Study: Relays per Address in Tor}
	\label{sec:tor}
	
	This section applies the capacity model of Section~\ref{ssec:capacity-model}
	and the threshold analysis of Section~\ref{ssec:threshold}, without
	modification, to a deployed Sybil defense whose sharing limit has been
	changed in practice. Its inputs are measured network data; its outputs
	are model-based costs.
	
	\subsection{Setting}
	\label{ssec:tor-setting}
	
	Tor~\cite{dingledine2004tor} relays are identities in the sense of
	Definition~\ref{def:identity}: each is listed separately in the network
	consensus. To raise the number of addresses an attacker must control,
	the directory authorities accept only a limited number of relays per
	IPv4 address~\cite{torprop109}. In the terms of
	Section~\ref{ssec:capacity-model}, an IPv4 address is a shareable
	resource instance and the limit is its capacity $m$. The limit was $2$
	until 2~February~2023, was raised to $4$, and was raised to $8$ at the
	end of June~2023~\cite{tor2023cap4,tor2023cap8}. We use the
	hidden-service directory (HSDir) role, in which positions are assigned
	per relay rather than per unit of bandwidth~\cite{torrendspec}, so that
	influence is count-based in the sense of Section~\ref{ssec:count-based}.
	Roles that are selected in proportion to bandwidth are resource-based
	(Section~\ref{ssec:resource-based}) and are not analyzed here.
	
	\subsection{Data and Model}
	\label{ssec:tor-data}
	
	The honest population is the median daily number of relays with the
	HSDir flag over the fourteen days before 2~February~2023 (19~January to
	1~February, all days present), taken from the public Tor Metrics
	data~\cite{tormetrics}: $N=2814$. One consensus from each period,
	obtained from the CollecTor archive~\cite{torcollector}, confirms the
	limits in force: at most $2$ relays per IPv4 address on 31~January,
	$4$ on 15~March, and $8$ on 15~July~2023. Only aggregate counts are
	used.
	
	For a target share $\tau$, the adversary needs $g_\tau$ relays
	(Section~\ref{ssec:threshold}), and under limit $k$ its minimum cost is
	$A(k)=K\lceil g_\tau/k\rceil+g_\tau c$
	(Lemma~\ref{lem:capacity-costs}), where $K$ is the cost of an address and
	$c$ the cost of operating one relay. The independent baseline is $k=1$,
	one address per relay, so that $A(1)/A(k)=R_\tau$. As in the model, the
	honest population is held at its pre-change size. We found no
	documented per-relay cost in the same unit as an address price, so
	costs are reported in units of $c$ as functions of $r=K/c$ rather than
	at an assumed price.
	
	\subsection{Results}
	\label{ssec:tor-results}
	
	With $\tau=1/4$, the adversary needs $g_\tau=938$ relays. The number
	of addresses this requires falls from $469$ under a limit of $2$ to
	$235$ under $4$ and $118$ under $8$. Table~\ref{tab:tor} gives the
	resulting costs, and Fig.~\ref{fig:tor} shows the reduction relative to
	a limit of $2$ over $0.25\le r\le 50$.
	
	\begin{table*}[t]
		\caption{Tor HSDir threshold with measured network size $N=2814$ and
			$\tau=1/4$ ($g_\tau=938$ relays): addresses required under each
			per-address limit $k$, reduction (red.) in attack cost relative to
			$k=2$, and overestimation factor $R_\tau=A(1)/A(k)$, for three
			price ratios $r=K/c$.}
		\label{tab:tor}
		\centering\small
		\begin{tabular}{rrrrrrrr}
			\hline
			& & \multicolumn{2}{c}{$r=1$} & \multicolumn{2}{c}{$r=5$} & \multicolumn{2}{c}{$r=20$}\\
			$k$ & addresses & red.\ (\%) & $R_\tau$ & red.\ (\%) & $R_\tau$ & red.\ (\%) & $R_\tau$\\
			\hline
			$1$ & $938$ & -- & $1.00$ & -- & $1.00$ & -- & $1.00$\\
			$2$ & $469$ & $0$ & $1.33$ & $0$ & $1.71$ & $0$ & $1.91$\\
			$4$ & $235$ & $16.6$ & $1.60$ & $35.6$ & $2.66$ & $45.4$ & $3.49$\\
			$8$ & $118$ & $24.9$ & $1.78$ & $53.5$ & $3.68$ & $68.0$ & $5.97$\\
			\hline
		\end{tabular}
	\end{table*}
	
	\begin{figure}[t]
		\centering
		\includegraphics[width=\columnwidth]{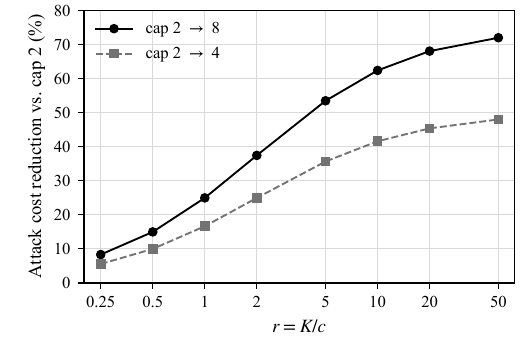}
		\caption{Reduction in the cost of reaching a quarter of the HSDir
			relays, relative to a limit of two relays per address, as a function
			of the price ratio $r=K/c$, with the measured network size $N=2814$.
			Values for $\tau=1/10$ and $\tau=1/2$ differ from those shown by less
			than $0.5$ percentage points. Both curves lie within $0.45$
			percentage points of the large-population limits in
			Proposition~\ref{prop:cap-limit}.}
		\label{fig:tor}
	\end{figure}
	
	Raising the limit from $2$ to $8$ lowers the cost of reaching the
	threshold by $24.9\%$, $53.5\%$, and $68.0\%$ at $r=1$, $5$, and $20$,
	and by between $8.3\%$ and $72.0\%$ over the whole range considered.
	Proposition~\ref{prop:cap-limit} gives the large-population limit
	$3r/(4(r+2))$ for this change, and $r/(2(r+2))$ for a rise from $2$ to
	$4$. At $g_\tau=938$ the exact reduction from $2$ to $8$ differs from
	its limit by $0.05$, $0.11$, and $0.15$ percentage points at the same
	three ratios. The limit equals one half at $r=4$. Over
	$\tau\in\{1/10,1/4,1/2\}$ and every ratio plotted in
	Fig.~\ref{fig:tor}, for both changes, the exact reduction differs from
	the limit by at most $0.45$ percentage points; the three shares differ
	from one another by less than $0.5$ percentage points. An
	estimate that assumes one address per relay overstates the cost under
	a limit of $8$ by a factor of $1.78$, $3.68$, and $5.97$ at $r=1$, $5$,
	and $20$; each value is at most the bound $(K+c)/(K/m+c)$ of
	Proposition~\ref{prop:overestimation}. At $r=20$ that factor is $5.97$,
	while one further relay placed in spare capacity is overestimated by
	$(K+c)/c=21$, which is greater than the capacity $8$
	(Proposition~\ref{prop:marginal}). The reduction grows with $r$, as
	Proposition~\ref{prop:cap-limit} states, that is, with the share of a
	relay's cost that lies in its address. Within the
	model, the scarcity of IPv4 addresses that motivated raising the
	limit~\cite{tor2023cap4,nothingtohide2023relays} is therefore the condition under which the
	change lowers the attack cost most.
	
	\subsection{Population Growth and Scope}
	\label{ssec:tor-scope}
	
	These are model-based costs on measured inputs, not measurements of an
	attack, and they do not establish that the limit was exploited. The
	reference calculation holds the honest population at its pre-change
	median. The count includes any adversarial relays already present. A
	single price per address is assumed, the adversary is assumed to obtain
	addresses freely, and bandwidth requirements, detection, and other
	admission rules are outside the comparison.
	
	The same cost formula shows how large a later population must be before
	the rise from a limit of $2$ to a limit of $8$ stops lowering the model
	cost of the quarter-share. For a large population that factor is
	$4(r+2)/(r+8)$: $4/3$, $28/13$, and $22/7$ at $r=1$, $5$, and $20$. At
	the measured February size, the cost under limit $8$ is no longer lower
	than the February cost under limit $2$ once the honest population
	reaches $3748$, $6052$, and $8833$ at those three ratios. The HSDir
	median over the fourteen days before 28~June~2023 (14--27~June, all days
	present) is $3866.5$. This count precedes the second change and is not
	attributed to it. After the limit had become $8$, the median over
	15--28~July, all fourteen days present, is $3921$, and over
	18--31~August it is $4336.5$. Relative to $2814$ these later medians are
	factors $1.39$ and $1.54$. They are not attributed to the limit either.
	Substituting either one for the February population, the cost under
	limit $8$ is higher than the February cost under limit $2$ at $r=1$, by
	$4.5\%$ and $15.6\%$, and it is still lower at $r=5$, by $35.2\%$ and
	$28.4\%$, and at $r=20$, by $55.5\%$ and $50.9\%$. At the later
	measured sizes, the cost under limit $8$ exceeds the February cost
	under limit $2$ only near $r=1$.
	
	\section{Discussion and Limitations}
	\label{sec:discussion}
	
	\subsection{Implications for System Design}
	\label{ssec:disc-design}
	\label{ssec:disc-interpretation}
	
	The cost that a requirement is intended to impose per identity need not
	equal the cost an adversary incurs to provision each identity. A
	mechanism may assign influence to each accepted identity while allowing
	the resources that support those identities to be shared. Identity
	amortization is therefore not a property of a resource in isolation,
	but of the interaction between sharing opportunities, provisioning
	requirements, and the cost model. Three implications for design follow.
	
	First, designers should identify the unit that a mechanism counts.
	Whether additional identities matter depends on whether influence is
	assigned per identity, per controlling entity, or per unit of some
	resource, so identity count alone does not characterize an adversary's
	influence.
	
	Second, designers should separate shareable costs from
	identity-specific ones and ask which resources honest participants and
	adversaries actually buy. Raising an identity-specific cost is paid on
	every identity. Raising a resource price falls on the adversary only in
	proportion to the instances it uses, and is regressive when the honest
	participant buys a different resource. A sharing cap below the onset
	removes the reduction without changing the single-identity price. None
	of these is universally preferable: admission, enforcement, and
	accessibility costs, and legitimate participants who need several
	identities, also matter.
	
	Third, identity counts should not be read as counts of independent
	participants. Sharing may, however, also help defenders: where shared
	resources are observable, they can support clustering of identities,
	as limits on accounts per phone number do~\cite{thomas2014dialing},
	although shared provisioning does not guarantee that the relationship
	is visible or that clusters reflect common control. The arrangement
	that lowers provisioning cost may thus also create signals that a
	defense can use.
	
	\subsection{Scope and Limitations}
	\label{ssec:disc-limitations}
	
	The analysis gives conditional results within an explicit provisioning
	model; it does not measure how often identity amortization occurs in
	deployed systems. Costs are in one model unit, the comparison is static
	(Section~\ref{ssec:def}), and
	the results depend on the independent baseline and on the granularity
	of resource instances. The catalogue results assume linear instance
	prices, whereas the general model admits bulk discounts, and the
	intervention comparison considers an honest participant with one
	identity. Identity-verification and admission costs, which may add
	identity-specific requirements, are not modeled. The two evaluations apply the model to published or measured
	inputs but do not measure attack costs. The phone catalogue uses prices
	from 2013--2014; Section~\ref{ssec:eval-scope} varies the capacity and
	the identity-specific cost on those prices and replaces only the code
	price with a current, platform-averaged value. The Tor case study reports costs at the pre-change
	population, and Section~\ref{ssec:tor-scope} compares them with later
	measured populations without attributing the change in population to
	the limit. The price ratio $r$ remains undocumented. The paper does not establish that every system with
	shareable resources is vulnerable to a successful Sybil attack.
	
	\subsection{Future Directions}
	\label{ssec:disc-future}
	
	A first direction is a matched current catalogue, with a code price, a
	SIM price, an identity-specific cost, and an observed reuse taken from
	one source, together with the per-number limits that services enforce
	and the costs adversaries incur. A second is a model of how honest
	operators add identities when a limit is raised. The later HSDir counts
	in Section~\ref{ssec:tor-scope} are not that model. Further
	directions are sharing across services, where a resource provisioned
	for one service also supports identities in others, which a
	single-system model does not capture; whether observable sharing is
	reliable evidence of common control, and under what conditions it can
	support clustering without merging independent participants; and the
	formal relation to false-name-proof mechanisms
	(Section~\ref{ssec:rw-positioning}), in which assumptions about the
	cost of additional identities affect the robustness of outcomes.
	
	\section{Conclusion}
	\label{sec:conclusion}
	
	A resource requirement intended to impose a cost per identity need not
	be paid per identity, and whether it is can be decided from prices and
	capacities alone. We defined identity amortization as the failure of
	the per-identity cost guarantee on which resource-based Sybil defenses
	rest, gave a test for when it occurs in any catalogue of resource types,
	the number of identities at which it begins, the largest error it
	causes, and the sharing cap that removes it, and compared defensive
	interventions against a re-optimizing adversary. On published
	phone-verification prices, a per-account estimate overstated the
	adversary's provisioning cost by a factor between $1.6$ and $4.3$, and
	the cheapest reported market price was below the least cost of
	verifying each account with its own number; the first strict
	reduction occurred where predicted, raising the price of the honest
	participant's resource left the adversary's cost unchanged, and a limit
	of three accounts per number removed the reduction. Bringing the error
	down to $1.2$ through the identity-specific cost instead would have
	raised the cost of an honest account $2.3$ times; the limit leaves it
	unchanged. For Tor, raising the limit
	on relays per address from two to eight lowered the modeled cost of
	reaching a quarter of the hidden-service directories by between $8\%$
	and $72\%$. These values lie within half a percentage point of the
	large-population limit, which depends on the price ratio and the two
	limits but not on the target share. Identity amortization does not by itself establish a
	successful Sybil attack, but evaluating Sybil resistance requires
	examining not only the nominal cost of a resource but also how much of
	it must be paid separately for each identity.
	
	\appendices
\section{Proofs}
\label{app:proofs}
	This appendix proves the results of Sections~\ref{sec:capacity} and~\ref{sec:security}; notation is as in the main text.
	
	\subsection*{Proof of Lemma~\ref{lem:capacity-costs} (Costs in the capacity model)}
	\begin{proof}
		Under independent provisioning every identity requires its own instance,
		so exactly $g$ instances are needed. Under shared provisioning, $g$
		identities on instances of capacity $m$ require at least $\lceil g/m\rceil$
		instances, and $\lceil g/m\rceil$ instances suffice by filling them in
		turn. The identity-specific cost $gc$ is the same in both arrangements.
		Every set of $g$ identities has the same cost, so these values are the
		count-level infima, each achieved by an explicit strategy.
	\end{proof}
	
	\subsection*{Proof of Proposition~\ref{prop:capacity-cost-reduction} (Cost reduction in the capacity model)}
	\begin{proof}
		(a) Subtracting the two costs, the term $gc$ cancels and
		$\Delta(g)=K(g-\lceil g/m\rceil)$. Because $m\ge1$, $\lceil g/m\rceil\le g$,
		so $\Delta(g)\ge0$; it is strictly positive exactly when
		$\lceil g/m\rceil<g$, which for integers $g,m\ge1$ holds exactly when $m>1$
		and $g>1$.
		(b) Dividing the two costs gives
		$\alpha(g)=(K\lceil g/m\rceil+gc)/(g(K+c))$. Since
		$\lceil g/m\rceil/g\ge1/m$, with equality exactly when $m$ divides $g$, and
		$\lceil g/m\rceil/g\to1/m$ as $g\to\infty$, the bound, the equality
		condition, and the limit follow.
	\end{proof}
	
	\subsection*{Proof of Proposition~\ref{prop:marginal} (Marginal overestimation)}
	\begin{proof}
		The shared cost of $g$ identities is $K\lceil g/m\rceil+gc$, so one
		further identity costs $K\bigl(\lceil(g+1)/m\rceil-\lceil g/m\rceil\bigr)+c$.
		For an integer $g$, the difference of ceilings is $0$ when $m\nmid g$ and
		$1$ when $m\mid g$. The independent cost of one identity is $K+c$, and
		dividing gives the two ratios. For $c>0$, $(K+c)/c>m$ if and only if
		$K>(m-1)c$.
	\end{proof}
	
	\subsection*{Proof of Proposition~\ref{prop:cap-limit} (Large-population cap increase)}
	\begin{proof}
		The shared costs at capacities $\ell$ and $k$ are $K\lceil g/\ell\rceil+gc$
		and $K\lceil g/k\rceil+gc$. Dividing both by $g$, and using
		$\lceil g/t\rceil/g\to1/t$ for each fixed integer $t\ge1$, the ratio of
		the two costs tends to $(K/\ell+c)/(K/k+c)$. If $c=0$, this limit is
		$k/\ell$, so the relative reduction tends to $(\ell-k)/\ell$. If $c>0$,
		the relative reduction tends to
		\[
		1-\frac{K/\ell+c}{K/k+c}
		=\frac{K(\ell-k)}{\ell(K+kc)}
		=\frac{r(\ell-k)}{\ell(r+k)}.
		\]
		The map $r\mapsto r/(r+k)$ has derivative $k/(r+k)^2>0$ for $r>0$, so it
		is strictly increasing, and its values lie in $(0,1)$. The stated limit
		therefore lies strictly between $0$ and $(\ell-k)/\ell$.
	\end{proof}
	
	\subsection*{Proof of Proposition~\ref{prop:onset} (Existence and onset)}
	\begin{proof}
		If $\rho^*=K_{\min}$, the per-identity floor gives an instance cost of at
		least $gK_{\min}$, so $\Cbshared(g)\ge g(K_{\min}+c)=\Cbind(g)$ for all
		$g$. Otherwise some type $i$ is amortizing. For such a type,
		$K_i/K_{\min}<m_i$ implies $g_i:=\lfloor K_i/K_{\min}\rfloor+1\le m_i$, so a
		single instance of type $i$ supports $g_i$ identities at instance cost
		$K_i<g_iK_{\min}$.
		
		Let $g$ be such that $K_i\ge gK_{\min}$ for every amortizing type. An
		instance of a non-amortizing type that supports $j\le m_i$ identities has
		instance cost $K_i\ge m_iK_{\min}\ge jK_{\min}$, and any instance costs at
		least $K_{\min}$. If a combination contains an instance of an amortizing
		type, that instance alone costs at least $gK_{\min}$; otherwise each
		instance costs at least $K_{\min}$ times the number of identities it
		supports. In both cases the instance cost is at least $gK_{\min}$, so
		$\Cbshared(g)\ge\Cbind(g)$. A strict reduction at $g$ therefore requires
		$g>K_i/K_{\min}$ for some amortizing type, that is, $g\ge g_i$, and hence
		$g\ge g^*$.
		
		Conversely, let $i$ be an amortizing type attaining $g^*$ and let $t$ be a
		type with $K_t=K_{\min}$. If $g^*\le g\le m_i$, one instance of type $i$
		supports $g$ identities at instance cost $K_i<g^*K_{\min}\le gK_{\min}$. If
		$g>m_i$, one instance of type $i$ together with $g-m_i$ instances of type
		$t$ supports $g$ identities at instance cost
		$K_i+(g-m_i)K_{\min}<m_iK_{\min}+(g-m_i)K_{\min}=gK_{\min}$, because
		$K_i<m_iK_{\min}$. In both cases $gc$ is the same under either arrangement,
		so $\Cbshared(g)<\Cbind(g)$.
	\end{proof}
	
	\subsection*{Proof of Corollary~\ref{cor:safe-cap} (Safe sharing cap)}
	\begin{proof}
		A cap $k$ replaces each $m_i$ by $\min(m_i,k)$ and leaves the prices, and
		hence $K_{\min}$, unchanged. A non-amortizing type remains non-amortizing,
		because $K_i\ge m_iK_{\min}\ge\min(m_i,k)K_{\min}$. An amortizing type has
		$K_i/K_{\min}<m_i$, so under the cap it is amortizing exactly when
		$kK_{\min}>K_i$. No type is amortizing under the cap if and only if
		$k\le\lfloor K_i/K_{\min}\rfloor$ for every amortizing type, that is,
		$k\le g^*-1$; by the existence part of the proposition, this is exactly
		when identity amortization occurs for no number of identities.
	\end{proof}
	
	\subsection*{Proof of Proposition~\ref{prop:error-bound} (Magnitude)}
	\begin{proof}
		The per-identity floor gives $\Cbshared(g)\ge g(\rho^*+c)$, and
		$\Cbind(g)=g(K_{\min}+c)$, which proves the bound. If $m_b\mid g$, then
		$g/m_b$ instances of type $b$ cost exactly $g\rho^*$, so the bound is
		attained. In general $\Cbshared(g)\le K_b\lceil g/m_b\rceil+gc$, so
		$\Cbshared(g)/g\to\rho^*+c$. Finally, $E\le K_{\min}/\rho^*=m_bK_{\min}/K_b
		\le m_b$, since $K_{\min}\le K_b$.
	\end{proof}
	
	\subsection*{Proof of Proposition~\ref{prop:overestimation} (Bound on attack-cost overestimation)}
	\begin{proof}
		Since $\lceil g_\tau/m\rceil\ge g_\tau/m$,
		\[
		\frac{g_\tau(K+c)}{K\lceil g_\tau/m\rceil+g_\tau c}
		\le\frac{g_\tau(K+c)}{Kg_\tau/m+g_\tau c}=\frac{K+c}{K/m+c}.
		\]
		Moreover, $K/m+c\ge(K+c)/m$, because $c\ge c/m$, so
		$(K+c)/(K/m+c)\le m$. Equality in the first inequality holds exactly when
		$\lceil g_\tau/m\rceil=g_\tau/m$, that is, $m\mid g_\tau$; equality in the
		second holds exactly when $(m-1)c=0$, that is, $m=1$ or $c=0$.
	\end{proof}
	
	\subsection*{Proof of Proposition~\ref{prop:interventions} (Defensive interventions)}
	\begin{proof}
		(a) Every combination pays $c$ per identity, so every cost, and hence the
		minimum, rises by $g\delta$; $H$ rises by $\delta$; and
		$dE/dc=(\rho^*-K_{\min})/(\rho^*+c)^2<0$.
		(b) The old optimal combination remains feasible and its cost rises by
		$\delta x_j$, and no combination becomes cheaper, so the minimum rises by at
		most $\delta x_j$. Before the change, $H=K_{\min}+c$. If $j$ is not the
		unique type of price $K_{\min}$, another type still has that price and $H$
		is unchanged. If it is, then $K'>K_j$ and the new honest cost is
		$\min(K_j+\delta,K')+c$, an increase of $\min(\delta,K'-K_j)$.
		(c) Immediate from the definition of $E$.
		(d) For every cap $k\ge1$ one identity still requires a single instance
		and the instance prices are unchanged, so $H$ is unchanged; the rest is the
		safe-cap corollary.
	\end{proof}

	\section*{Acknowledgments}
	This contribution results from the Horizon Europe funded project RETIME - Urban Adaptation and Alert Solutions for a TIMEly (Re) Action, co-funded by the European Union under Grant Agreement No. 101147113.

	\bibliographystyle{IEEEtran}
	\bibliography{references}
	
\end{document}